\documentclass[letterpaper,10pt,conference]{ieeeconf}  
\IEEEoverridecommandlockouts                              
\usepackage{cite}
\usepackage{amsmath,amssymb,amsfonts}
\usepackage{mathrsfs}
\usepackage{graphicx}
\usepackage{xcolor}
\let\labelindent\relax
\usepackage{enumitem}
\usepackage[normalem]{ulem}
\makeatletter
\let\NAT@parse\undefined
\makeatother
\usepackage[
    colorlinks=true,
    linkcolor=blue,
    citecolor=blue,
    urlcolor=blue
]{hyperref}

\usepackage{flushend}

\graphicspath{{figures/}}

\definecolor{myred}{rgb}{0.7,0.1,0.16}

\newtheorem{theorem}{Theorem}
\newtheorem{proposition}{Proposition}

\newtheorem{lemma}{Lemma}

\newtheorem{assumption}{Assumption}

\newcommand{\R}{\mathbb{R}}

\catcode`\@=11
\def\downparenfill{$\m@th\braceld\leaders\vrule\hfill\bracerd$}
\def\overparen#1{\mathop{\vbox{\ialign{##\crcr\crcr \noalign{\kern0.4ex}
\downparenfill\crcr\noalign{\kern0.4ex\nointerlineskip}
$\hfil\displaystyle{#1}\hfil$\crcr}}}\limits}
\catcode`\@=12

\title{\LARGE \bf
Partial-Scan-and-Move Source Seeking for Mobile Robots
}

\author{Bo Wang and Ishvar Sitaldin 
\thanks{Bo Wang and Ishvar Sitaldin are with the Department of Mechanical Engineering, The City College of New York, The City University of New York, New York, NY 10031, USA (e-mail: bwang1@ccny.cuny.edu). }
}

\begin{document}

\maketitle
\thispagestyle{empty}
\pagestyle{empty}

\begin{abstract}
This paper presents a partial-scan-and-move strategy for source seeking with a mobile robot equipped with an offset scalar sensor. At each robot position, the sensor collects source field measurements while the robot rotates. Instead of requiring a complete circular scan before every move, we ask when the measurements collected over only part of the circle are already sufficient to determine the next action. We develop a gradient estimation method for partial scans together with a confidence set that accounts for measurement noise and local field variation. The robot uses this confidence set to decide whether it is close enough to the source or has enough information to move in a descent direction. We show that, under suitable conditions, each decision can be made within a prescribed partial scan and that the robot reaches a desired neighborhood of the source in finitely many moves with high probability. 
\end{abstract}

\section{Introduction}\label{sec:introduction}

Mobile robot source seeking concerns the problem of locating an unknown source using only local measurements of a spatially varying signal field \cite{zhang2007source,cochran2009nonholonomic}. This problem arises in a wide range of applications, including environmental monitoring, hazardous-material detection, search and rescue, and autonomous exploration. Depending on the application, the measured field may represent electromagnetic intensity, chemical concentration, light intensity, acoustic energy, temperature, or other physical quantities. In many such tasks, the robot does not have prior knowledge of the source location or a complete map of the field. It must therefore infer how to move toward the source from measurements collected along its trajectory. This makes source seeking a fundamental problem at the intersection of sensing, estimation, and motion control for autonomous robots.

A large body of work on mobile robot source seeking is based on extremum-seeking control \cite{krstic2000stability,DurrStankovicEbenbauerJohansson2013,Suttner2023TAC}. In these methods, the robot deliberately excites its motion and uses the resulting variation in the measured field value to infer a direction toward the source, without requiring a model of the field or knowledge of the source location. Both deterministic schemes based on periodic probing signals \cite{wang2023underactuated,wang2026nonholonomicsourceseekingtorque,wang2025extremum} and stochastic variants \cite{liu2010stochastic,lin2017stochastic} have been developed for nonholonomic mobile robots. Other source-seeking approaches estimate the spatial gradient from measurements collected at different locations, either by a single moving robot or by multiple cooperating agents~\cite{Burian1996,Ogren2004}. There are also navigation methods that steer the robot toward an extremum without explicitly reconstructing the field gradient, for example through suitably designed steering or sliding-mode laws \cite{matveev2011navigation}. Despite their different implementations, many of these approaches share a common idea: spatial information about an unknown field must be generated from measurements taken at different sensor locations and then converted into a motion direction.

The present work is closely related to source-seeking methods that estimate spatial gradient information directly from scalar measurements collected at different locations. An early example is the AUV gradient-search method in \cite{Burian1996}, where vehicle maneuvers, including circular motions, are used to obtain measurements to estimate the local gradient. Least-squares gradient estimation has also been used in cooperative source seeking, where measurements from multiple robots at different positions are combined to estimate the gradient at the formation center \cite{LiGuo2012}, and in extremum-seeking schemes based on first-order or recursive least-squares estimation \cite{Hunnekens2014,Zengin2020}. More recently, \cite{JamesKavanaugh2026} used a servo-actuated sensor to explore the local field around a nonholonomic vehicle and extract derivative information for source seeking. These works show that spatially distributed scalar measurements can provide useful gradient information. In contrast, this paper asks a different question: whether a circular scan can be stopped early once sufficient information is available to move toward the source.

In this paper, we develop a \textit{scan-and-move strategy} for source seeking with an offset scalar sensor. At each robot position, the robot rotates in place and collects measurements along a circular arc before moving toward the source. A complete circular scan can be inefficient because sufficient directional information may be available well before a full rotation is completed. We therefore allow the robot to stop scanning as soon as the collected measurements provide enough information for reliable motion. To the best of our knowledge, partial circular scans and their stopping conditions have not been systematically studied for source seeking. The main contributions are threefold: (i) we characterize the limitations of gradient estimation from fixed-radius measurements and show the additional error introduced by first-harmonic estimation under partial scans; (ii) we develop a second-order harmonic least-squares estimator and an anytime-valid confidence set for the local gradient, which provides a stopping rule for the scan; and (iii) we give conditions guaranteeing that each scan terminates before a full rotation and prove finite-move convergence to a prescribed gradient tolerance with high probability.


\section{Partial-Scan Source Seeking}\label{sec:problem}

\subsection{Scan-and-Move Source-Seeking Setup}

Consider a planar mobile robot with unicycle kinematics. Let $p=[x~y]^\top\in\R^2$ and $\theta\in\mathbb{R}$ denote the robot position and heading. We assume that a scalar sensor is mounted at a fixed distance $\rho>0$ from the robot center along its heading direction. With $\varphi(\theta)=[\cos\theta~\sin\theta]^\top$, the unicycle kinematics and the position of the scalar sensor are
\begin{equation}
\dot p=v\varphi(\theta),\quad
\dot\theta=\omega,\quad
p^{\rm s}=p+\rho\varphi(\theta),
\label{eq:model}
\end{equation}
where $v$ and $\omega$ are the translational and angular velocities, respectively, and $p^{\rm s}\in\R^2$ denotes the sensor position.

We assume that the robot operates in an unknown static scalar signal field $F:\R^2\to\R$, which can be measured only through the onboard sensor. The motion is organized into \textit{scan-and-move episodes} indexed by $k\in\mathbb{Z}_{\geq0}$. During the scan phase of episode $k$, the robot center is held fixed at $p_k$. Let $\theta_k$ denote the robot heading at the beginning of the episode, and let $\alpha_{k,i}$ denote the relative scan bearing of the $i$th measurement with respect to this initial heading. Thus, the sensor heading at the $i$th measurement is $\theta_k+\alpha_{k,i}$, and the measurement is
\begin{equation}
\mathsf{y}_{k,i}=F\!\left(p_k+\rho\varphi(\theta_k+\alpha_{k,i})\right)+w_{k,i},
\label{eq:measurement}
\end{equation}
where $w_{k,i}$ denotes measurement noise. The relative scan bearing $\alpha_{k,i}$ is known to the robot, whereas the absolute heading $\theta_k$ need not be known.

We make the following assumption on the field $F$.
\begin{assumption}\label{ass:main}
The scalar field $F$ has a unique global minimizer $p^\star\in\R^2$, referred to as the source, and $p^\star$ is the only stationary point of $F$. The initial sublevel set $\{p\in\R^2:F(p)\leq F(p_0)\}$ is compact. The field satisfies $F\in C^3$. Its gradient is uniformly bounded and Lipschitz continuous, and its third-order derivatives are uniformly bounded. Specifically, there exist known constants $G,L,M_3>0$ such that, for all $p,\bar p\in\R^2$, $\|\nabla F(p)\|\leq G$, $\|\nabla F(p)-\nabla F(\bar p)\|\leq L\|p-\bar p\|$, and $\|D^3F(p)\|\leq M_3$. Finally, the initial field-value gap satisfies $F(p_0)-F(p^\star)\leq\Delta_0$ for some known constant $\Delta_0>0$.
\end{assumption}

Let $\mathcal F_{k,i}$ denote the information generated by all relative scan bearings and measurements up to sample $i$ of episode $k$, including previous episodes. We make the following assumption on the measurement noise and the execution of the scan-and-move procedure.
\begin{assumption}\label{ass:noise-motion}
For each episode $k$ and sample $i$, the measurement noise $w_{k,i}$ is conditionally zero-mean given the information available before the $i$th measurement, i.e., $\mathbb{E}[w_{k,i}\;|\;\mathcal F_{k,i-1}]=0$. It is also conditionally sub-Gaussian with known variance proxy $\sigma^2$, i.e., $\mathbb{E}[e^{\lambda w_{k,i}}\;|\;\mathcal F_{k,i-1}]\leq\exp(\lambda^2\sigma^2/2)$ for all $\lambda\in\mathbb{R}$. The relative scan bearing $\alpha_{k,i}$ is $\mathcal F_{k,i-1}$-measurable; thus, it may depend on previous measurements but is chosen before the $i$th measurement is taken. During each scan, the robot holds its center fixed and, after the scan, executes the commanded rotation and translation.
\end{assumption}

\textit{Problem Statement:} Given prescribed $\varepsilon>0$ and $\delta\in(0,1)$, design a scan-and-move strategy using only scalar field measurements to approach the source $p^\star$ and, with probability at least $1-\delta$, terminate after finitely many episodes at a point $p_K$ satisfying $\|\nabla F(p_K)\|\leq\varepsilon$. 

Here, $\varepsilon$ specifies the desired terminal accuracy, measured by the gradient norm, while $\delta$ bounds the probability that the guarantee fails.

\subsection{Limits of Fixed-Radius Gradient Estimation}

Before considering partial scans, we first ask how accurately the gradient can be recovered even from a complete noiseless scan. To isolate the information limitation of fixed-radius sensing, we consider the broader class $\mathscr{F}_3(M_3)$ of fields that are $C^3$ on a neighborhood of the closed sensing disk and satisfy $\|D^3F\|\leq M_3$ there.

For a fixed scan center $p$ and sensing radius $\rho$, define the complete noiseless circular observation
\begin{equation}
\mathcal{Y}_{\rho,p}(F):=\left\{\left(\alpha,F\!\left(p+\rho\varphi(\alpha)\right)\right):
\alpha\in[0,2\pi)\right\}.
\label{eq:circular-observation}
\end{equation}

\begin{lemma}[Fixed-radius gradient estimation]\label{lem:ambiguity}
For every estimator $\mathcal{E}$ that maps $\mathcal{Y}_{\rho,p}(F)$ to a vector in $\R^2$,
\begin{equation}
\sup_{F\in\mathscr{F}_3(M_3)}\left\|\mathcal{E}\!\left(\mathcal{Y}_{\rho,p}(F)\right)-\nabla F(p)\right\|\geq\frac{M_3\rho^2}{6}.
\label{eq:ambiguity}
\end{equation}
\end{lemma}

\begin{proof}
Fix a unit vector $a_0\in\R^2$ and let $a=(M_3\rho^2/6)a_0$. For $z\in\R^2$, define $h(z)=a^\top z(1-\|z\|^2/\rho^2)$ and consider the two fields $F_+(p+z)=h(z)$ and $F_-(p+z)=-h(z)$. If $\|z\|=\rho$, then $h(z)=0$, and hence $\mathcal{Y}_{\rho,p}(F_+)=\mathcal{Y}_{\rho,p}(F_-)$. However, $\nabla F_+(p)=a$ and $\nabla F_-(p)=-a$.

Direct differentiation gives $\|D^3h\|\leq6\|a\|/\rho^2=M_3$, so both $F_+$ and $F_-$ belong to $\mathscr{F}_3(M_3)$. Since they produce identical observations, $\mathcal E$ returns the same estimate, denoted by $\hat g$, for both fields. Therefore,
\begin{equation}
\max\{\|\hat g-a\|,\|\hat g+a\|\}\geq\|a\|=\frac{M_3\rho^2}{6},
\end{equation}
which proves \eqref{eq:ambiguity}.
\end{proof}

Lemma~\ref{lem:ambiguity} shows that fixed-radius sensing over a general $C^3$ field class has an unavoidable worst-case gradient error of at least order $\rho^2$, even when the \textit{complete} noiseless circular observation is available. This limitation is due to the sensing geometry itself: different fields can produce exactly the same measurements on the sensing circle while having different gradients at its center. Improving uniformly beyond this scale therefore requires information beyond a single sensing circle, such as additional radial measurements.

\subsection{First-Harmonic Regression under Partial Scans}

Lemma~\ref{lem:ambiguity} shows that fixed-radius sensing has an unavoidable worst-case gradient-error lower bound of order $\rho^2$. We now ask whether a simple gradient estimate based on a \textit{partial scan} introduces a larger error.

For the analysis below, we use the episode-local frame, so $\alpha$ directly denotes the scan bearing. For a scan centered at $p$, the second-order Taylor expansion gives
\begin{equation}
\begin{aligned}
F\!\left(p+\rho\varphi(\alpha)\right)=&F(p)+\rho\nabla F(p)^\top\varphi(\alpha)\\
&+\frac{\rho^2}{2}\varphi(\alpha)^\top\nabla^2F(p)\varphi(\alpha)+r_p(\alpha),
\end{aligned}
\label{eq:taylor-expansion}
\end{equation}
where $r_p(\alpha)$ is the third-order Taylor remainder satisfying $|r_p(\alpha)|\leq M_3\rho^3/6$. Using $\nabla F(p)^\top\varphi(\alpha)=F_x(p)\cos\alpha+F_y(p)\sin\alpha$ and expanding the quadratic term yields
\begin{equation}
\begin{aligned}
F\!\left(p+\rho\varphi(\alpha)\right)&=c_0(p)+\rho F_x(p)\cos\alpha+\rho F_y(p)\sin\alpha\\
&\;+c_{\rm c}(p)\cos2\alpha+c_{\rm s}(p)\sin2\alpha+r_p(\alpha),
\end{aligned}
\label{eq:harmonic}
\end{equation}
where $c_0(p)=F(p)+\rho^2(F_{xx}(p)+F_{yy}(p))/4$, $c_{\rm c}(p)=\rho^2(F_{xx}(p)-F_{yy}(p))/4$, and $c_{\rm s}(p)=\rho^2F_{xy}(p)/2$.

Equation~\eqref{eq:harmonic} separates the local field variation into angular harmonics. The gradient components appear in the coefficients of $\cos\alpha$ and $\sin\alpha$, while curvature contributes to the constant and second-harmonic terms. Define
\begin{equation}
\zeta^1(p)=
\begin{bmatrix}
c_0(p)\\
\rho F_x(p)\\
\rho F_y(p)
\end{bmatrix},
\quad
\zeta^2(p)=
\begin{bmatrix}
c_{\rm c}(p)\\
c_{\rm s}(p)
\end{bmatrix}.
\label{eq:harmonic-parameters}
\end{equation}
The gradient can be recovered from $\zeta^1(p)$ as
\begin{equation}
\nabla F(p)=\frac{1}{\rho}E\zeta^1(p),\quad 
E=\begin{bmatrix}
0&1&0\\
0&0&1
\end{bmatrix}.
\label{eq:gradient-from-zeta1}
\end{equation}

Consider episode $k$ after $n$ noiseless measurements at relative scan bearings $\alpha_{k,1},\ldots,\alpha_{k,n}$. Define
\begin{equation}
\Phi^1_{k,n}=
\begin{bmatrix}
1&\cos\alpha_{k,1}&\sin\alpha_{k,1}\\
\vdots&\vdots&\vdots\\
1&\cos\alpha_{k,n}&\sin\alpha_{k,n}
\end{bmatrix}
\in\R^{n\times3},
\label{eq:Phi1}
\end{equation}
and
\begin{equation}
\Phi^2_{k,n}=
\begin{bmatrix}
\cos2\alpha_{k,1}&\sin2\alpha_{k,1}\\
\vdots&\vdots\\
\cos2\alpha_{k,n}&\sin2\alpha_{k,n}
\end{bmatrix}
\in\R^{n\times2}.
\label{eq:Phi2}
\end{equation}
Let
\begin{equation}
\mathsf{Y}_{k,n}=
\begin{bmatrix}
\mathsf{y}_{k,1}\\
\vdots\\
\mathsf{y}_{k,n}
\end{bmatrix},
\quad
r_{k,n}=
\begin{bmatrix}
r_{p_k}(\alpha_{k,1})\\
\vdots\\
r_{p_k}(\alpha_{k,n})
\end{bmatrix}.
\label{eq:stacked-data}
\end{equation}
Then \eqref{eq:harmonic} gives
\begin{equation}
\mathsf{Y}_{k,n}=\Phi^1_{k,n}\zeta^1(p_k)+\Phi^2_{k,n}\zeta^2(p_k)+r_{k,n}.
\label{eq:stacked-harmonic-model}
\end{equation}

A natural method is to ignore the second-harmonic and Taylor-remainder terms and fit the constant and first harmonics by ordinary least squares. If $\Phi^1_{k,n}$ has full column rank, the ordinary least-squares estimate is
\begin{align}
\hat\zeta^{\,1}_{k,n}&=\arg\min_{\zeta\in\R^3}\|\mathsf{Y}_{k,n}-\Phi^1_{k,n}\zeta\|^2 \notag\\
&=\big((\Phi^1_{k,n})^\top\Phi^1_{k,n}\big)^{-1}(\Phi^1_{k,n})^\top\mathsf{Y}_{k,n}.
\label{eq:first-harmonic-OLS}
\end{align}
The corresponding gradient estimate is
\begin{equation}
\hat g^{\,1}_{k,n}=\frac{1}{\rho}E\hat\zeta^{\,1}_{k,n}.
\label{eq:first-harmonic-gradient}
\end{equation}

For a \textit{partial scan}, the following result characterizes the error of this first-harmonic gradient estimate.

\begin{proposition}[Curvature leakage]\label{prop:leakage}
Suppose that the measurements are noiseless and $\Phi^1_{k,n}$ has full column rank. Then
\begin{equation}
\hat g^{\,1}_{k,n}-\nabla F(p_k)=\rho B_{k,n}
\begin{bmatrix}
(F_{xx}(p_k)-F_{yy}(p_k))/4\\
F_{xy}(p_k)/2
\end{bmatrix}
+\varepsilon_{k,n},
\label{eq:leakage}
\end{equation}
where
\begin{equation}
B_{k,n}=E\big((\Phi^1_{k,n})^\top\Phi^1_{k,n}\big)^{-1}(\Phi^1_{k,n})^\top\Phi^2_{k,n},
\label{eq:Bkn}
\end{equation}
and
\begin{equation}
\varepsilon_{k,n}=\frac{1}{\rho}E\big((\Phi^1_{k,n})^\top\Phi^1_{k,n}\big)^{-1}(\Phi^1_{k,n})^\top r_{k,n}.
\label{eq:first-harmonic-remainder}
\end{equation}
Moreover,
\begin{equation}
\|\varepsilon_{k,n}\|\leq\frac{M_3\rho^2\sqrt n}{6}\left\|E\big((\Phi^1_{k,n})^\top\Phi^1_{k,n}\big)^{-1}(\Phi^1_{k,n})^\top\right\|.
\label{eq:first-harmonic-remainder-bound}
\end{equation}
\end{proposition}

\begin{proof}
Substituting \eqref{eq:stacked-harmonic-model} into \eqref{eq:first-harmonic-OLS} gives
\begin{align}
\hat\zeta^{\,1}_{k,n}-\zeta^1(p_k)
&=
\big((\Phi^1_{k,n})^\top\Phi^1_{k,n}\big)^{-1}
(\Phi^1_{k,n})^\top\Phi^2_{k,n}\zeta^2(p_k) \notag\\
&\;\;+
\big((\Phi^1_{k,n})^\top\Phi^1_{k,n}\big)^{-1}
(\Phi^1_{k,n})^\top r_{k,n}.
\end{align}
Premultiplying by $E/\rho$ gives
\begin{equation}
\hat g^{\,1}_{k,n}-\nabla F(p_k)
=
\frac{1}{\rho}B_{k,n}\zeta^2(p_k)
+\varepsilon_{k,n}.
\end{equation}
Since
\begin{equation}
\zeta^2(p_k)
=
\rho^2
\begin{bmatrix}
(F_{xx}(p_k)-F_{yy}(p_k))/4\\
F_{xy}(p_k)/2
\end{bmatrix},
\end{equation}
\eqref{eq:leakage} follows. Finally, $|r_{p_k}(\alpha_{k,i})|\leq M_3\rho^3/6$ implies $\|r_{k,n}\|\leq\sqrt n\,M_3\rho^3/6$, which gives \eqref{eq:first-harmonic-remainder-bound}.
\end{proof}

For an equally spaced \textit{complete-circle scan} with $n\geq5$, the first and second angular harmonics are orthogonal, so $(\Phi^1_{k,n})^\top\Phi^2_{k,n}=0$ and hence $B_{k,n}=0$. The curvature term in \eqref{eq:leakage} then vanishes, leaving only the $O(\rho^2)$ Taylor-remainder error. For a generic \textit{partial scan}, however, $(\Phi^1_{k,n})^\top\Phi^2_{k,n}\neq0$. For a fixed, well-conditioned scan geometry, the second-harmonic curvature terms therefore produce an additional error of order $O(\rho)$.

Thus, the difficulty with first-harmonic regression under a partial scan is not simply that fewer measurements are available. The loss of harmonic orthogonality causes local curvature to be interpreted as gradient information. This $O(\rho)$ term is larger than the $\rho^2$ lower-bound scale in Lemma~\ref{lem:ambiguity}, which motivates estimating the first and second harmonics jointly in the next section.

\section{Gradient Estimation from Partial Scans}\label{sec:estimation}

\subsection{Second-Order Harmonic Regression}

To remove the curvature leakage identified in Proposition~\ref{prop:leakage}, we retain both the first and second angular harmonics in the regression. The remaining deterministic model mismatch is the third-order Taylor remainder, whose magnitude is $O(\rho^3)$. Define
\begin{equation}
z(\alpha)=
\begin{bmatrix}
\cos\alpha\\
\sin\alpha\\
\cos2\alpha\\
\sin2\alpha
\end{bmatrix},
\qquad
\xi(p)=
\begin{bmatrix}
\rho F_x(p)\\
\rho F_y(p)\\
c_{\rm c}(p)\\
c_{\rm s}(p)
\end{bmatrix}.
\label{eq:second-order-parameter}
\end{equation}
Then \eqref{eq:harmonic} can be written as
\begin{equation}
    F(p+\rho\varphi(\alpha))=c_0(p)+z(\alpha)^\top\xi(p)+r_p(\alpha).
\end{equation}
Since $\nabla F$ is Lipschitz and $F\in C^3$, $\|\nabla^2F(p)\|\leq L$. Together with $\|\nabla F(p)\|\leq G$ and the definitions of $c_{\rm c}$ and $c_{\rm s}$, this gives
\begin{equation}
\|\xi(p)\|\leq S_\xi:=\sqrt{\rho^2G^2+\frac{\rho^4L^2}{4}}.
\label{eq:Sxi}
\end{equation}
Thus, the four harmonic coefficients collected in $\xi(p)$ have a known bound.

The offset $c_0(p_k)$ is different: adding a constant to $F$ changes $c_0$ but does not change the gradient or the source. Hence $c_0(p_k)$ need \textit{not} be bounded and should not be regularized. We separate this unrestricted offset from the bounded harmonic coefficients by anchoring the regression at the first scan bearing. Define
\begin{equation}
a_k=c_0(p_k)+z(\alpha_{k,1})^\top\xi(p_k),\quad
\vartheta_k=
\begin{bmatrix}
a_k\\
\xi(p_k)
\end{bmatrix}\in\R^5,
\label{eq:anchored-parameter}
\end{equation}
and
\begin{equation}
\psi_{k,1}=e_1,\quad
\psi_{k,i}=
\begin{bmatrix}
1\\
z(\alpha_{k,i})-z(\alpha_{k,1})
\end{bmatrix},\quad i\geq2,
\label{eq:anchored-regressor}
\end{equation}
where $e_1=[1~0~0~0~0]^\top$. One gets the reparameterization: for every $i\ge 1$,
\begin{equation}
\mathsf{y}_{k,i}=\psi_{k,i}^\top\vartheta_k+r_{p_k}(\alpha_{k,i})+w_{k,i}, 
\label{eq:anchored-model}
\end{equation}
where $|r_{p_k}(\alpha_{k,i})|\leq b_\rho:={M_3\rho^3}/{6}$.

For $\lambda>0$, writing $\vartheta=[a~\xi^\top]^\top$, we estimate $\vartheta_k$ by partially regularized least squares:
\begin{equation}
\hat\vartheta_{k,n}=\arg\min_{\vartheta\in\R^5}
\left\{\sum_{i=1}^{n}\big(\mathsf{y}_{k,i}-\psi_{k,i}^\top\vartheta\big)^2+\lambda\|\xi\|^2\right\}.
\label{eq:regularized-ls}
\end{equation}
Thus, the offset $a$ is left unregularized, while the four harmonic coefficients in $\xi$ are regularized. Define
\begin{equation}
\Lambda_\lambda:=\operatorname{diag}(0,\lambda I_4),\quad
V_{k,n}:=\Lambda_\lambda+\sum_{i=1}^{n}\psi_{k,i}\psi_{k,i}^\top.
\label{eq:information-matrix}
\end{equation}
The solution of \eqref{eq:regularized-ls} is
\begin{equation}
\boxed{\hat\vartheta_{k,n}=V_{k,n}^{-1}\sum_{i=1}^{n}\psi_{k,i}\mathsf{y}_{k,i}.}
\label{eq:second-order-estimator}
\end{equation}
Because $\psi_{k,1}=e_1$, $V_{k,1}=\operatorname{diag}(1,\lambda I_4)\succ0$, so $V_{k,n}$ is nonsingular for every $n\geq1$.
Thus, the estimate is defined from the first measurement onward; when little scan information is available, the uncertainty is reflected by the confidence set introduced below rather than by requiring a full-rank scan.

The first two components of $\xi(p_k)$ are $\rho\nabla F(p_k)$. Define
\begin{equation}
G_\rho:=\frac{1}{\rho}\begin{bmatrix}0&I_2&0_{2\times2}\end{bmatrix},\quad
\boxed{\hat g_{k,n}:=G_\rho\hat\vartheta_{k,n}.}
\label{eq:gradient-estimate}
\end{equation}
Then $\hat g_{k,n}$ is the gradient estimate at $p_k$ in the episode-local frame.

\subsection{Gradient Confidence Sets}

A point estimate alone is not sufficient for deciding when to stop a partial scan: the stopping time depends on the measurements collected so far. We therefore construct a \textit{confidence set} that contains the true gradient simultaneously for all within-episode sample sizes. Define
\begin{equation}
P_{k,n}:=G_\rho V_{k,n}^{-1}G_\rho^\top.
\label{eq:Pkn}
\end{equation}
The matrix $P_{k,n}$ is the inverse-information matrix projected onto the two gradient coordinates and determines the shape of the gradient uncertainty.

Let $\delta_k\in(0,1)$ be an episode-wise failure probability and define
\begin{equation}
\begin{aligned}
\beta_{k,n}:={}&\sigma\sqrt{2\log\frac{4}{\delta_k}}
+\sigma\sqrt{2\log\!\left(\frac{2\det(V_{k,n})^{1/2}}{\delta_k\det(V_{k,1})^{1/2}}\right)}\\
&+\sqrt{\lambda}\,S_\xi+b_\rho\sqrt n.
\end{aligned}
\label{eq:beta}
\end{equation}
The first two terms in $\beta_{k,n}$ account for measurement noise, the third for regularization bias, and the last for the third-order Taylor remainder. The corresponding gradient confidence set is
\begin{equation}
\boxed{\mathcal C_{k,n}:=\left\{g\in\R^2:(g-\hat g_{k,n})^\top P_{k,n}^{-1}(g-\hat g_{k,n})\leq\beta_{k,n}^2\right\}.}
\label{eq:gradient-confidence-set}
\end{equation}
Thus, $\hat g_{k,n}$ is the center of the ellipsoid, $P_{k,n}$ describes its directional shape, and $\beta_{k,n}$ sets its overall size.

\begin{theorem}[Anytime gradient confidence set]\label{thm:anytime-confidence}
Under Assumptions~\ref{ass:main} and~\ref{ass:noise-motion}, for every episode $k$,
\begin{equation}
\mathbb P\!\left(
\nabla F(p_k)\in\mathcal C_{k,n}\text{ for all }n\geq1
\,\middle|\,\mathcal F_{k,0}
\right)\geq1-\delta_k,
\label{eq:anytime-confidence}
\end{equation}
where the gradient is expressed in the episode-local frame. 
Since the guarantee holds for all $n$ simultaneously, the confidence set remains valid at any stopping time $\tau_k$.
\end{theorem}

\begin{proof}
Suppress the episode index and write $r_i=r_{p_k}(\alpha_{k,i})$. From \eqref{eq:anchored-model} and \eqref{eq:second-order-estimator},
\begin{equation}
V_n(\hat\vartheta_n-\vartheta)
=\sum_{i=1}^{n}\psi_iw_i+\sum_{i=1}^{n}\psi_ir_i-\Lambda_\lambda\vartheta.
\label{eq:parameter-error}
\end{equation}
The three terms on the right are, respectively, the measurement-noise contribution, the Taylor-remainder contribution, and the regularization bias.

Since $V_n\succeq V_1$, $\|\psi_1w_1\|_{V_n^{-1}}\leq|w_1|$. The conditional sub-Gaussian assumption gives $|w_1|\leq\sigma\sqrt{2\log(4/\delta_k)}$ with probability at least $1-\delta_k/2$. Applying a self-normalized martingale inequality \cite{NIPS2011_e1d5be1c} to the predictable regressors $\psi_i$, $i\geq2$, gives, simultaneously for all $n$,
\begin{equation}
\left\|\sum_{i=2}^{n}\psi_iw_i\right\|_{V_n^{-1}}
\leq
\sigma\sqrt{2\log\!\left(\frac{2\det(V_n)^{1/2}}{\delta_k\det(V_1)^{1/2}}\right)}
\label{eq:self-normalized-bound}
\end{equation}
with probability at least $1-\delta_k/2$.

Let $\Psi_n$ have rows $\psi_i^\top$ and let $r_n=[r_1~\cdots~r_n]^\top$. Since $V_n\succeq\Psi_n^\top\Psi_n$ and $|r_i|\leq b_\rho$,
\begin{equation}
\left\|\sum_{i=1}^{n}\psi_ir_i\right\|_{V_n^{-1}}\leq\|r_n\|\leq b_\rho\sqrt n.
\label{eq:remainder-confidence-bound}
\end{equation}
Moreover, since only $\xi(p_k)$ is regularized,
\begin{equation}
\|\Lambda_\lambda\vartheta\|_{V_n^{-1}}\leq\sqrt{\lambda}\,\|\xi(p_k)\|\leq\sqrt{\lambda}\,S_\xi.
\label{eq:regularization-confidence-bound}
\end{equation}
Combining the three bounds yields, with probability at least $1-\delta_k$,
\begin{equation}
\|\hat\vartheta_n-\vartheta\|_{V_n}\leq\beta_{k,n}\qquad\text{for all }n\geq1.
\label{eq:parameter-confidence}
\end{equation}

Finally, the gradient is the linear projection $G_\rho\vartheta$. For every $e\in\R^5$,
\begin{equation}
(G_\rho e)^\top\big(G_\rho V_n^{-1}G_\rho^\top\big)^{-1}(G_\rho e)\leq e^\top V_ne.
\label{eq:projected-ellipsoid}
\end{equation}
Applying \eqref{eq:projected-ellipsoid} to $e=\hat\vartheta_n-\vartheta$ gives \eqref{eq:anytime-confidence}. Because the event holds simultaneously for all $n$, it also holds at any stopping time $\tau_k$.
\end{proof}

The confidence set \eqref{eq:gradient-confidence-set} therefore combines the information in the current partial scan with measurement noise, regularization, and the Taylor remainder. Section~\ref{sec:strategy} uses this set to determine when the scan contains enough information to guarantee a descent direction or to declare approximate stationarity.

\section{Partial-Scan-and-Move Strategy}\label{sec:strategy}

\subsection{Descent Direction and Scan Decisions}

The confidence set in \eqref{eq:gradient-confidence-set} contains the gradients that remain consistent with the measurements at the prescribed confidence level. This subsection addresses two questions: whether the current scan already provides enough information to stop, and, if motion is justified, which direction guarantees descent for every gradient in the confidence set. Define
\begin{equation}
\gamma^-_{k,n}:=\min_{g\in\mathcal C_{k,n}}\|g\|,\qquad
\gamma^+_{k,n}:=\max_{g\in\mathcal C_{k,n}}\|g\|.
\label{eq:gamma-minus-plus}
\end{equation}
Thus, $\gamma^-_{k,n}$ and $\gamma^+_{k,n}$ are the smallest and largest gradient magnitudes allowed by the current confidence set. If $\gamma^+_{k,n}\leq\varepsilon$, then every admissible gradient has norm at most $\varepsilon$; on the confidence event, this establishes approximate \textit{stationarity}. Otherwise, motion requires a direction that decreases the field for every gradient in the set.

\begin{proposition}[Descent from a gradient confidence set]\label{prop:confidence-descent}
Let $\mathcal C\subset\R^2$ be nonempty, compact, and convex, and let
\begin{equation}
g^\sharp=\operatorname{proj}_{\mathcal C}(0):=\arg\min_{g\in\mathcal C}\|g\|.
\label{eq:gradient-projection}
\end{equation}
If $g^\sharp\neq0$, then
\begin{equation}
d^\sharp:=-\frac{g^\sharp}{\|g^\sharp\|}
\label{eq:descent-direction}
\end{equation}
satisfies
\begin{equation}
g^\top d^\sharp\leq-\min_{z\in\mathcal C}\|z\|,\quad \forall g\in\mathcal C.
\label{eq:common-descent}
\end{equation}
Moreover, for any $\eta\in(0,1)$, there exists a unit direction $d$ satisfying $g^\top d\leq-\eta\max_{z\in\mathcal C}\|z\|$ for every $g\in\mathcal C$ if and only if
\begin{equation}
\min_{g\in\mathcal C}\|g\|\geq\eta\max_{g\in\mathcal C}\|g\|.
\label{eq:relative-descent-condition}
\end{equation}
\end{proposition}

\begin{proof}
The projection optimality condition gives $(g-g^\sharp)^\top g^\sharp\geq0$ for every $g\in\mathcal C$. Hence $g^\top d^\sharp\leq-\|g^\sharp\|$, and $\|g^\sharp\|=\min_{g\in\mathcal C}\|g\|$. Conversely, for any unit $d$,
$\min_{g\in\mathcal C}(-g^\top d)\leq-(g^\sharp)^\top d\leq\|g^\sharp\|$.
Thus, \eqref{eq:descent-direction} gives the largest worst-case descent margin, and \eqref{eq:relative-descent-condition} follows.
\end{proof}

The parameter $\eta\in(0,1)$ controls how much confidence is required before the robot is allowed to move. A larger $\eta$ requires a stronger guaranteed decrease of the field and therefore typically leads to a longer scan.

\textbf{Scan decision:} Applied to $\mathcal C_{k,n}$, Proposition~\ref{prop:confidence-descent} gives the following scan decision. Continue scanning until either
\begin{equation}
\boxed{\gamma^+_{k,n}\leq\varepsilon
\quad\text{or}\quad
\gamma^-_{k,n}\geq\eta\gamma^+_{k,n}.}
\label{eq:scan-stopping-rule}
\end{equation}
The first condition declares approximate stationarity and takes priority if both conditions hold. The second condition means that the current scan is sufficiently informative to guarantee a descent direction with \textit{relative margin} $\eta$. In that case, let $g^\sharp_{k,n}=\operatorname{proj}_{\mathcal C_{k,n}}(0)$ and use
\begin{equation}
d_{k,n}:=-\frac{g^\sharp_{k,n}}{\|g^\sharp_{k,n}\|}.
\label{eq:episode-descent-direction}
\end{equation}
By Theorem~\ref{thm:anytime-confidence} and Proposition~\ref{prop:confidence-descent}, when the true gradient belongs to the confidence set,
\begin{equation}
\nabla F(p_k)^\top d_{k,n}\leq-\gamma^-_{k,n}.
\label{eq:true-gradient-descent}
\end{equation}
Thus, the scan stops for one of two reasons: either the gradient is known to be small enough, or the measurements identify a direction that is guaranteed to decrease the field\footnote{Both vectors in \eqref{eq:true-gradient-descent} are expressed in the episode-local frame, so the required rotation is relative to the heading at the beginning of the episode and no absolute heading is needed.}. 

\subsection{Stopping a Partial Scan}

The scan decision in \eqref{eq:scan-stopping-rule} is useful only if it can be reached without completing a full $2\pi$ rotation. We therefore prescribe a maximum scan angle $\bar\alpha\in(0,2\pi)$ and show when the available information is guaranteed to be sufficient within that arc.

\textbf{Scan strategy:} Fix an integer $m\geq5$. In every episode, take the first measurement at the current heading and use the candidate relative bearings
\begin{equation}
\alpha_{k,i}=\frac{i-1}{m-1}\bar\alpha,\qquad i=1,\ldots,m.
\label{eq:partial-arc-schedule}
\end{equation}
The scan is checked after every measurement and stops as soon as \eqref{eq:scan-stopping-rule} holds. Thus, $m$ is the maximum number of prescribed samples, and $\bar\alpha<2\pi$ is the maximum rotation within an episode.\hfill $\blacktriangle$

Let $\Psi_m(\bar\alpha)$ be the matrix whose rows are the anchored regressors in \eqref{eq:anchored-regressor} evaluated at the bearings in \eqref{eq:partial-arc-schedule}, and let
$V_m(\bar\alpha)=\Lambda_\lambda+\Psi_m^\top\Psi_m$ and
$V_1=\operatorname{diag}(1,\lambda I_4)$. Define
\begin{equation}
\kappa_m(\bar\alpha):=\lambda_{\min}\!\left(\frac{1}{m}\Psi_m^\top\Psi_m\right),\;
D_m(\bar\alpha):=\frac{\det V_m(\bar\alpha)}{\det V_1}.
\label{eq:partial-design-quantities}
\end{equation}
The quantity $\kappa_m(\bar\alpha)$ measures how well the prescribed bearings excite the five regression coefficients: a small value indicates a poorly conditioned scan geometry. The determinant ratio $D_m(\bar\alpha)$ measures the information growth appearing in the noise bound of Theorem~\ref{thm:anytime-confidence}. For $\bar\alpha>0$ and $m\geq5$, the five-parameter harmonic design is full rank at the distinct bearings in \eqref{eq:partial-arc-schedule}, so $\kappa_m(\bar\alpha)>0$. Both quantities depend only on the prescribed scan geometry and can be computed offline. For $\delta\in(0,1)$, define
\begin{equation}
\begin{aligned}
\overline\beta_m(\bar\alpha,\delta):={}&\sigma\sqrt{2\log\frac{4}{\delta}}
+\sigma\sqrt{2\log\!\left(\frac{2D_m(\bar\alpha)^{1/2}}{\delta}\right)}\\
&+\sqrt{\lambda}\,S_\xi+b_\rho\sqrt m,
\end{aligned}
\label{eq:beta-bar}
\end{equation}
and
\begin{equation}
\overline\varrho_m(\bar\alpha,\delta):=
\frac{\overline\beta_m(\bar\alpha,\delta)}
{\rho\sqrt{m\kappa_m(\bar\alpha)}}.
\label{eq:partial-radius-bound}
\end{equation}
The bound $\overline\beta_m$ collects the effects of measurement noise, regularization, and the Taylor remainder after at most $m$ samples. The quantity $\overline\varrho_m$ converts this bound into a Euclidean bound on the gradient uncertainty. In particular, after the $m$ prescribed samples, every gradient in the confidence set lies within $\overline\varrho_m$ of its center. Hence $\overline\varrho_m$ can be evaluated before the experiment to check whether a proposed partial-scan design is sufficiently informative.

\begin{theorem}[Partial-Scan Termination]\label{thm:partial-scan-termination}
Fix $\varepsilon>0$, $\eta\in(0,1)$, and $\delta_k\in(0,1)$. If the scan parameters satisfy
\begin{equation}
\overline\varrho_m(\bar\alpha,\delta_k)\leq\frac{(1-\eta)\varepsilon}{4},
\label{eq:partial-resolution-condition}
\end{equation}
then, on the event that
$\nabla F(p_k)\in\mathcal C_{k,n}$ for all $n\geq1$, episode $k$ satisfies one of the two conditions in \eqref{eq:scan-stopping-rule} by sample $m$. Hence $\tau_k\leq m$, and the total scan angle is at most $\bar\alpha<2\pi$.
\end{theorem}

\begin{proof}
Let $\varrho_{k,m}:=\beta_{k,m}\sqrt{\lambda_{\max}(P_{k,m})}$ denote the Euclidean radius of the gradient confidence ellipsoid at sample $m$. From \eqref{eq:partial-design-quantities},
$V_{k,m}\succeq m\kappa_m(\bar\alpha)I_5$, and hence $\lambda_{\max}(P_{k,m})\leq{1}/{(\rho^2m\kappa_m(\bar\alpha))}$.
Together with \eqref{eq:beta}--\eqref{eq:partial-radius-bound}, this gives
$\varrho_{k,m}\leq\overline\varrho_m(\bar\alpha,\delta_k)$.

Let $g_k=\nabla F(p_k)$ in the episode-local frame. On the confidence event, both $g_k$ and every $g\in\mathcal C_{k,m}$ are within $\varrho_{k,m}$ of $\widehat g_{k,m}$. Thus every $g\in\mathcal C_{k,m}$ satisfies $\|g-g_k\|\leq2\varrho_{k,m}$, which yields
\begin{equation}
\gamma^+_{k,m}\leq\|g_k\|+2\varrho_{k,m},\;
\gamma^-_{k,m}\geq\max\{0,\|g_k\|-2\varrho_{k,m}\}.
\label{eq:gamma-radius-bounds}
\end{equation}

If $\|g_k\|\leq\varepsilon-2\varrho_{k,m}$, then \eqref{eq:gamma-radius-bounds} gives $\gamma^+_{k,m}\leq\varepsilon$, so the first stopping condition holds. Otherwise,
$\|g_k\|>\varepsilon-2\varrho_{k,m}$. Using
$\varrho_{k,m}\leq(1-\eta)\varepsilon/4$ gives $(1-\eta)\|g_k\|>2(1+\eta)\varrho_{k,m}$,
and therefore \eqref{eq:gamma-radius-bounds} implies
$\gamma^-_{k,m}\geq\eta\gamma^+_{k,m}$. Hence the second stopping condition holds. In either case, the scan must stop by sample $m$.
\end{proof}

Condition~\eqref{eq:partial-resolution-condition} is an offline design requirement: it asks that the gradient uncertainty after the prescribed partial scan be small enough to distinguish approximate stationarity from a reliable descent direction. For fixed $\varepsilon$, $\eta$, and $\delta_k$, it can be checked directly when choosing $m$, $\bar\alpha$, and $\lambda$.

\subsection{Finite-Move Convergence}

\textbf{Move strategy:} When the second condition in \eqref{eq:scan-stopping-rule} holds at the stopping time $\tau_k$, rotate toward $d_k:=d_{k,\tau_k}$ and translate by
\begin{equation}
\ell_k:=\frac{\gamma^-_{k,\tau_k}}{L},
\label{eq:move-length}
\end{equation}
where $L$ is the Lipschitz constant in Assumption~\ref{ass:main}. The new center is denoted by $p_{k+1}=p_k+\ell_k d_k$. If the first condition in \eqref{eq:scan-stopping-rule} holds, the algorithm terminates at $p_k$.\hfill $\blacktriangle$

\begin{theorem}[Finite-Move Convergence]\label{thm:finite-move}
Let
\begin{equation}
K_{\max}:=\left\lceil\frac{2L\Delta_0}{\eta^2\varepsilon^2}\right\rceil,
\quad
\delta_k:=\frac{\delta}{K_{\max}+1}
\label{eq:episode-confidence-allocation}
\end{equation}
for $k=0,\ldots,K_{\max}$.
Suppose the scan parameters $(\bar\alpha,m)$ satisfy \eqref{eq:partial-resolution-condition} for the choice of $\delta_k$ in \eqref{eq:episode-confidence-allocation}. Then, with probability at least $1-\delta$, every episode stops after at most $m$ measurements and a scan angle no greater than $\bar\alpha<2\pi$, and the algorithm terminates after at most $K_{\max}$ translational moves at a point $p_K$ satisfying
\begin{equation}
\|\nabla F(p_K)\|\leq\varepsilon.
\label{eq:terminal-gradient}
\end{equation}
\end{theorem}

\begin{proof}
On the event that the true gradient belongs to the confidence set throughout episode $k$, every nonterminal episode satisfies \eqref{eq:true-gradient-descent}. By the $L$-Lipschitz continuity of $\nabla F$ and \eqref{eq:move-length},
\begin{equation}
F(p_{k+1})\leq F(p_k)-\frac{(\gamma^-_{k,\tau_k})^2}{2L}.
\label{eq:field-decrease}
\end{equation}
Before termination, $\gamma^+_{k,\tau_k}>\varepsilon$ and $\gamma^-_{k,\tau_k}\geq\eta\gamma^+_{k,\tau_k}$, so every move decreases $F$ by more than $\eta^2\varepsilon^2/(2L)$. Since $F(p_0)-F(p^\star)\leq\Delta_0$, there can be at most $K_{\max}$ such moves. At the terminal episode, $\gamma^+_{K,\tau_K}\leq\varepsilon$; because the true gradient belongs to $\mathcal C_{K,\tau_K}$, \eqref{eq:terminal-gradient} follows.

At most $K_{\max}+1$ episode-wise confidence events are used. By Theorem~\ref{thm:anytime-confidence}, each fails with probability at most $\delta_k$ conditionally on the information available at the beginning of that episode. Therefore, the union bound and \eqref{eq:episode-confidence-allocation} imply that all required confidence events hold with probability at least $1-\delta$. The partial-scan statement then follows from Theorem~\ref{thm:partial-scan-termination}.
\end{proof}

Theorem \ref{thm:finite-move} gives the complete scan-and-move guarantee. Each nonterminal episode produces a field decrease bounded away from zero, so the initial field-value gap $\Delta_0$ limits the number of moves. The allocation of the prescribed failure probability $\delta$ across at most $K_{\max}+1$ episodes gives the overall probability guarantee, while Theorem~\ref{thm:partial-scan-termination} ensures that every decision is reached within a partial scan. Thus, with probability at least $1-\delta$, the strategy terminates after finitely many partial scans and moves at an $\varepsilon$-stationary point.

\section{Numerical Results}\label{sec:numerical}

We illustrate the proposed strategy on a planar unicycle moving in the smooth field
\begin{equation}
F(x,y)=a_1\log\cosh\!\left(\frac{q_1}{\ell_1}\right)+a_2\log\cosh\!\left(\frac{q_2}{\ell_2}\right),
\label{eq:simulation-field}
\end{equation}
where
\begin{equation}
\begin{bmatrix}
q_1\\q_2
\end{bmatrix}
=
R_\phi^\top\left(
\begin{bmatrix}
x\\y
\end{bmatrix}
-p^\star\right),\quad
R_\phi=
\begin{bmatrix}
\cos\phi & -\sin\phi\\
\sin\phi & \cos\phi
\end{bmatrix}.
\label{eq:simulation-rotation}
\end{equation}
Here, $p^\star=[15~10]^\top$ m, $\phi=25^\circ$, $a_1=6$, $a_2=4$, $\ell_1=35$ m, and $\ell_2=18$ m. The parameters $\ell_1$ and $\ell_2$ set the spatial scales along the two principal directions of the field. The field has a unique stationary point and global minimizer $p^\star$ and satisfies Assumption~\ref{ass:main}. In particular, the bounds used in the simulation are
$G=\sqrt{(a_1/\ell_1)^2+(a_2/\ell_2)^2}$,
$L=\max\{a_1/\ell_1^2,a_2/\ell_2^2\}$, and
$M_3=\frac{4}{3\sqrt{3}}\max\{a_1/\ell_1^3,a_2/\ell_2^3\}$.

The parameters are selected as $\rho=3$ m, $\lambda=10^{-5}$, $\sigma=5\times10^{-4}$, $\varepsilon=0.06$, $\eta=0.3$, $\delta=0.05$, $\bar\alpha=240^\circ$, and $m=25$. Independent zero-mean Gaussian noise with standard deviation $\sigma$ is added to each scalar measurement. The five initial states $(x,y,\theta)$ are $(-65,25,15^\circ)$, $(-35,-45,-40^\circ)$, $(10,70,100^\circ)$, $(75,55,170^\circ)$, and $(80,-30,-110^\circ)$. For these parameters, the offline partial-scan condition \eqref{eq:partial-resolution-condition} is satisfied.

The unicycle dynamics in \eqref{eq:model} are simulated directly. During each scan, $v=0$ and the robot rotates through the prescribed relative bearings while the offset sensor collects measurements. After every measurement, the second-order estimate, the confidence set, and the quantities $\gamma^-_{k,n}$ and $\gamma^+_{k,n}$ are updated, and the scan decision \eqref{eq:scan-stopping-rule} is checked. If motion is selected, the robot first rotates toward $d_{k,\tau_k}$ and then translates with $\omega=0$ over the distance $\ell_k$ in \eqref{eq:move-length}. We use $\omega_{\rm scan}=0.8$ rad/s, $\omega_{\rm turn}=1.2$ rad/s, and $v_{\rm move}=4$ m/s.

\begin{figure}
    \centering
    \includegraphics[width=0.75\linewidth]{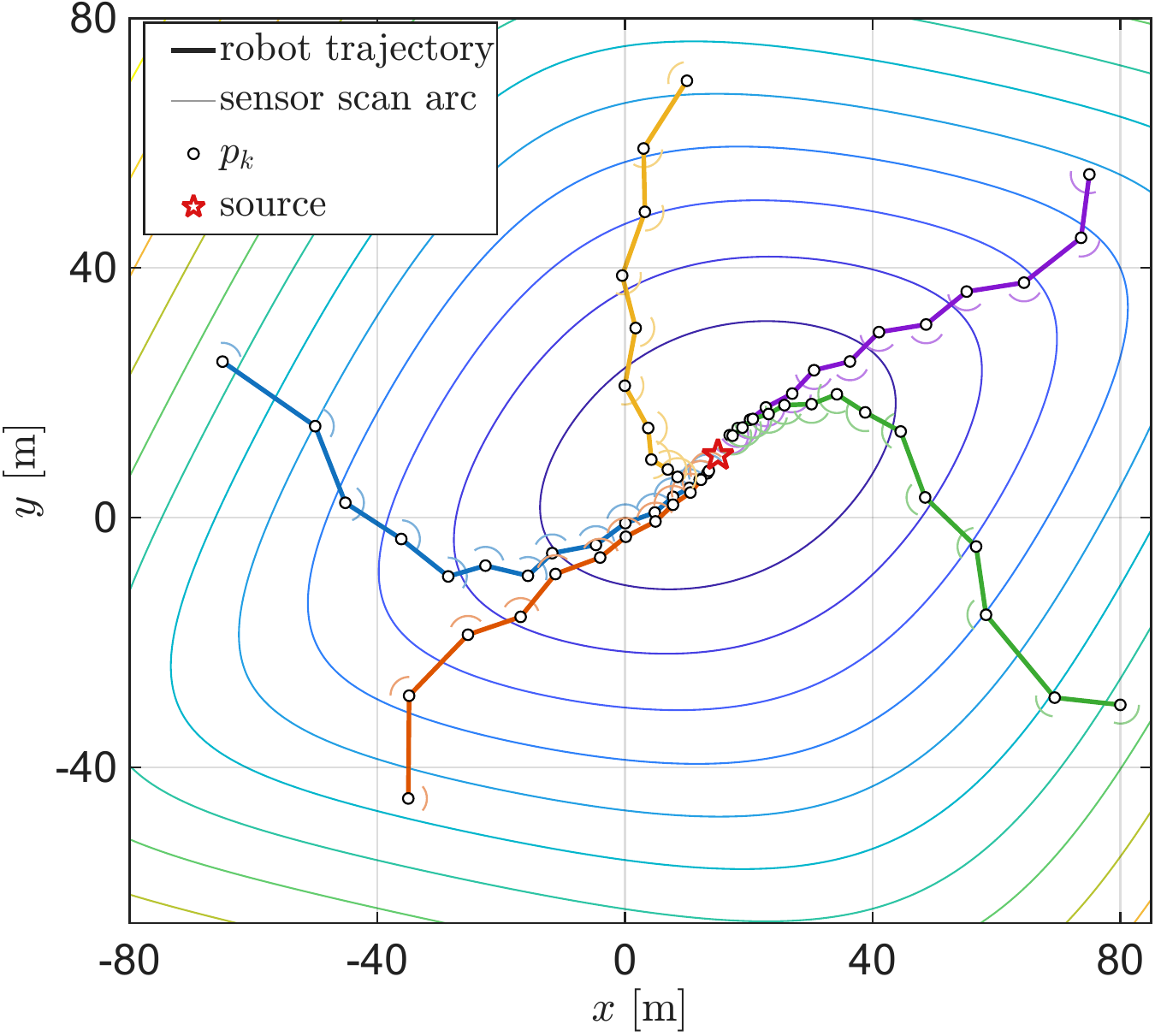}
    \caption{Partial-scan-and-move source seeking from different initial conditions. }
    \label{fig:partial-scan-convergence}
\end{figure}

Figure~\ref{fig:partial-scan-convergence} shows that the trajectories from different initial positions and headings approach the same source. Each open circle is an episode center $p_k$, and each light arc is the actual path of the offset sensor while the robot center is held fixed during that scan. The arc lengths vary because the scan is stopped as soon as one of the two conditions in \eqref{eq:scan-stopping-rule} is met; hence a full rotation is not required. After a nonterminal scan, the robot moves along the direction obtained from the current gradient confidence set. As the source is approached, the move lengths decrease, and each run terminates when the confidence set establishes $\gamma^+_{k,\tau_k}\leq\varepsilon$. The results illustrate that the proposed strategy uses only the partial scan needed to make the current decision while driving the robot toward an $\varepsilon$-stationary neighborhood of the source.

\section{Conclusion}\label{sec:conclusion}

This paper studied mobile robot source seeking with an offset scalar sensor using partial circular scans. We showed that fixed-radius sensing limits gradient estimation and that first-harmonic regression can suffer additional curvature-induced error under partial scans. A second-order harmonic model was therefore used to construct an anytime-valid gradient confidence set, which determines when to stop scanning and whether to declare approximate stationarity or select a descent direction. The resulting strategy guarantees finite-move convergence to the prescribed gradient tolerance with high probability, while avoiding a full rotation under suitable scan conditions.

\bibliographystyle{IEEEtran}
\bibliography{mybibfile}

\begin{thebibliography}{10}
\providecommand{\url}[1]{#1}
\csname url@samestyle\endcsname
\providecommand{\newblock}{\relax}
\providecommand{\bibinfo}[2]{#2}
\providecommand{\BIBentrySTDinterwordspacing}{\spaceskip=0pt\relax}
\providecommand{\BIBentryALTinterwordstretchfactor}{4}
\providecommand{\BIBentryALTinterwordspacing}{\spaceskip=\fontdimen2\font plus
\BIBentryALTinterwordstretchfactor\fontdimen3\font minus \fontdimen4\font\relax}
\providecommand{\BIBforeignlanguage}[2]{{%
\expandafter\ifx\csname l@#1\endcsname\relax
\typeout{** WARNING: IEEEtran.bst: No hyphenation pattern has been}%
\typeout{** loaded for the language `#1'. Using the pattern for}%
\typeout{** the default language instead.}%
\else
\language=\csname l@#1\endcsname
\fi
#2}}
\providecommand{\BIBdecl}{\relax}
\BIBdecl

\bibitem{zhang2007source}
C.~Zhang, D.~Arnold, N.~Ghods, A.~Siranosian, and M.~Krstic, ``Source seeking with non-holonomic unicycle without position measurement and with tuning of forward velocity,'' \emph{Systems \& control letters}, vol.~56, no.~3, pp. 245--252, 2007.

\bibitem{cochran2009nonholonomic}
J.~Cochran and M.~Krstic, ``Nonholonomic source seeking with tuning of angular velocity,'' \emph{IEEE Transactions on Automatic Control}, vol.~54, no.~4, pp. 717--731, 2009.

\bibitem{krstic2000stability}
M.~Krsti{\'c} and H.-H. Wang, ``Stability of extremum seeking feedback for general nonlinear dynamic systems,'' \emph{Automatica}, vol.~36, no.~4, pp. 595--601, 2000.

\bibitem{DurrStankovicEbenbauerJohansson2013}
H.-B. D{\"u}rr, M.~S. Stankovi{\'c}, C.~Ebenbauer, and K.~H. Johansson, ``{L}ie bracket approximation of extremum seeking systems,'' \emph{Automatica}, vol.~49, no.~6, pp. 1538--1552, 2013.

\bibitem{Suttner2023TAC}
R.~Suttner, ``Extremum-seeking control for a class of mechanical systems,'' \emph{IEEE Transactions on Automatic Control}, vol.~68, no.~2, pp. 1200--1207, 2023.

\bibitem{wang2023underactuated}
B.~Wang, S.~G. Nersesov, H.~Ashrafiuon, P.~Naseradinmousavi, and M.~Krsti{\'c}, ``Underactuated source seeking by surge force tuning: Theory and boat experiments,'' \emph{IEEE Transactions on Control Systems Technology}, vol.~31, no.~4, pp. 1649--1662, Jul. 2023.

\bibitem{wang2026nonholonomicsourceseekingtorque}
\BIBentryALTinterwordspacing
B.~Wang, ``Nonholonomic source seeking by torque tuning: Local and semi-global feedbacks,'' 2026. [Online]. Available: \url{https://arxiv.org/abs/2607.02458}
\BIBentrySTDinterwordspacing

\bibitem{wang2025extremum}
B.~Wang, H.~Ashrafiuon, and S.~G. Nersesov, ``Extremum seeking control for antenna pointing via symmetric product approximation,'' \emph{IFAC-PapersOnLine}, vol.~59, no.~30, pp. 869--874, 2025.

\bibitem{liu2010stochastic}
S.-J. Liu and M.~Krstic, ``Stochastic source seeking for nonholonomic unicycle,'' \emph{Automatica}, vol.~46, no.~9, pp. 1443--1453, 2010.

\bibitem{lin2017stochastic}
J.~Lin, S.~Song, K.~You, and M.~Krstic, ``Stochastic source seeking with forward and angular velocity regulation,'' \emph{Automatica}, vol.~83, pp. 378--386, 2017.

\bibitem{Burian1996}
E.~Burian, D.~R. Yoerger, A.~Bradley, and H.~Singh, ``Gradient search with autonomous underwater vehicles using scalar measurements,'' in \emph{Proc. IEEE Symposium on Autonomous Underwater Vehicle Technology}, 1996, pp. 86--98.

\bibitem{Ogren2004}
P.~{\"O}gren, E.~Fiorelli, and N.~E. Leonard, ``Cooperative control of mobile sensor networks: Adaptive gradient climbing in a distributed environment,'' \emph{IEEE Transactions on Automatic Control}, vol.~49, no.~8, pp. 1292--1302, 2004.

\bibitem{matveev2011navigation}
A.~S. Matveev, H.~Teimoori, and A.~V. Savkin, ``Navigation of a unicycle-like mobile robot for environmental extremum seeking,'' \emph{Automatica}, vol.~47, no.~1, pp. 85--91, 2011.

\bibitem{LiGuo2012}
S.~Li and Y.~Guo, ``Distributed source seeking by cooperative robots: All-to-all and limited communications,'' in \emph{Proceedings of the IEEE International Conference on Robotics and Automation}, 2012, pp. 1107--1112.

\bibitem{Hunnekens2014}
B.~G.~B. Hunnekens, M.~A.~M. Haring, N.~van~de Wouw, and H.~Nijmeijer, ``A dither-free extremum-seeking control approach using 1st-order least-squares fits for gradient estimation,'' in \emph{Proceedings of the 53rd IEEE Conference on Decision and Control}, 2014, pp. 2679--2684.

\bibitem{Zengin2020}
N.~Zengin and B.~Fidan, ``Adaptive extremum seeking using recursive least squares,'' \emph{arXiv preprint arXiv:2003.03891}, 2020.

\bibitem{JamesKavanaugh2026}
D.~James-Kavanaugh, P.~McNamee, Q.~Wang, and Z.~N. Ahmadabadi, ``Servos for local map exploration onboard nonholonomic vehicles for extremum seeking,'' \emph{IEEE Transactions on Control Systems Technology}, vol.~34, no.~5, pp. 2733--2748, 2026.

\bibitem{NIPS2011_e1d5be1c}
Y.~Abbasi-yadkori, D.~P\'{a}l, and C.~Szepesv\'{a}ri, ``Improved algorithms for linear stochastic bandits,'' in \emph{Advances in Neural Information Processing Systems}, J.~Shawe-Taylor, R.~Zemel, P.~Bartlett, F.~Pereira, and K.~Weinberger, Eds., vol.~24.\hskip 1em plus 0.5em minus 0.4em\relax Curran Associates, Inc., 2011.

\end{thebibliography}

\end{document}